\documentclass[twoside]{article}

\usepackage[round]{natbib}

\usepackage{fullpage}  
\usepackage{natbib} 
\usepackage{authblk} 

\usepackage[utf8]{inputenc}
\usepackage[english]{babel}
\usepackage{graphicx} 
\usepackage{amsmath}  
\usepackage{hyperref} 
\usepackage{fullpage}
\usepackage{booktabs} 
\usepackage[ruled]{algorithm2e} 

\usepackage{bm}
\usepackage{apxproof}
\usepackage{xcolor}
\usepackage{enumitem}
\usepackage{natbib}

\newcommand{\aX}{\mathcal{X}}
\newcommand{\reals}{\mathbf{R}}
\newcommand{\aV}{\mathcal{V}}
\newcommand{\aB}{\mathcal{B}}
\newcommand{\dbid}{\tilde{b}}
\newcommand{\rbid}{\bar{b}}
\newcommand{\ds}{\displaystyle}
\newcommand{\dom}{\mbox{supp}\ }
\newcommand{\dis}{\pi}
\newcommand{\rai}{\mu}
\newcommand{\vdis}{\bm{\dis}}
\newcommand{\vrai}{\bm{\rai}}
\newcommand{\vb}{\mathbf{b}}
\newcommand{\mbtw}{\mathcal{W}}
\newcommand{\mbtl}{\mathcal{L}}
\newcommand{\Vw}{V_w}
\newcommand{\Vl}{V_{\ell}}
\newcommand{\vzero}{\mathbf{0}}
\newcommand{\core}{\mathcal{C}}
\newcommand{\bicore}{\mathcal{C}_2}
\newcommand{\vp}{\mathbf{p}}
\newcommand{\vd}{\mathbf{d}}
\newcommand{\vx}{\mathbf{x}}
\newcommand{\vlambda}{\bm{\lambda}}

\theoremstyle{plain} 
\newtheorem{theorem}{Theorem}

\theoremstyle{plain} 
\newtheorem{proposition}{Proposition}

\theoremstyle{plain} 
\newtheorem{lemma}{Lemma}

\theoremstyle{definition} 
\newtheorem{definition}{Definition}

\theoremstyle{definition} 
\newtheorem{property}{Property}

\theoremstyle{definition} 
\newtheorem{example}{Example}

\theoremstyle{remark} 
\newtheorem*{remark}{Remark}

\DeclareMathOperator*{\argmax}{arg\,max}

\title{Bidder Feedback in First-Price Auctions for Video Advertising}

\author[1]{S\'{e}bastien Lahaie}
\author[1]{Benjamin Schaeffer}
\author[1]{Yuanjun Zhou}
\affil[1]{Google}

\date{\today}

\begin{document}
\maketitle


\begin{abstract}
In first-price auctions for display advertising, exchanges typically communicate the “minimum-bid-to-win” to bidders after the auction as feedback for their bidding algorithms. For a winner, this is the second-highest bid, while for losing bidders it is the highest bid. In this paper we investigate the generalization of this concept to general combinatorial auctions, motivated by the domain of video advertising. In a video pod auction, ad slots during an advertising break in a video stream are auctioned all at once, under several kinds of allocation constraints such as a constraint on total ad duration. We cast the problem in terms of computing bid updates (discounts and raises) that maintain the optimality of the current allocation. Our main result characterizes the set of joint bid updates with this property as the core of an associated bicooperative game. In the case of the assignment problem—a special case of video pod auctions—we provide a linear programming characterization of this bicooperative core. Our characterization leads to several candidates for a generalized minimum-bid-to-win. Drawing on video pod auction data from a real ad exchange, we perform an empirical analysis to understand the bidding dynamics they induce and their convergence properties.
\end{abstract}

\section{Introduction}

In recent years the display advertising industry has seen a rapid shift from second-price auctions to first-price auctions as the predominant mechanism for allocating and pricing ads. Several well-known exchanges started rolling out first-price auctions in 2017, and by 2019 the industry-wide move was complete~\citep{sluis17,sluis19}. The transition was driven by several factors, such as calls for enhanced transparency and the increasingly popular practice of header bidding, which is based on first-price auctions. On the demand side, advertisers have been compelled to adapt to this new environment. Whereas second-price auctions are truthful, under first-price auctions advertisers have to apply and optimize bid shading algorithms to avoid overpaying for their ad placements~\citep{benes18}.

To help advertisers calibrate their bids for future auctions, and to ease the transition from the second-price format, first-price ad exchanges provide price feedback to bidders. For bid shading algorithms, an important component of this feedback is the ``minimum-bid-to-win'' field reported by several platforms, including Google Ad Exchange. In a standard first-price auction for a single advertising slot, the meaning of this field is intuitive and straightforward: the minimum-bid-to-win for the winner is the second-highest bid, and for losing bidders it is the highest bid.\footnote{See for instance the definition of the {minimum\_bid\_to\_win} field in the OpenRTB protocol;~\url{https://developers.google.com/authorized-buyers/rtb/openrtb-guide##bidfeedback}.} Advertisers can use this information to create models of competing bid distributions for ad inventory and optimize their bid shading algorithms, streamlining the process towards market equilibrium.

In this paper we investigate how to generalize the concept of minimum-bid-to-win to general combinatorial auctions, motivated by pod auctions used in video advertising. An ad pod is a commercial break of multiple ads within online video content, similar to traditional linear television. Video content that shows ad pods includes streaming channels on connected TVs, live-streamed events (e.g., sports), and mobile apps for TV channels. In contrast to standard TV ad breaks which are planned ahead of time, ad pods are filled in real-time by a combination of pre-booked reservation ads and programmatic ads that bid for placement into the pod.

Figure~\ref{fig:video-pod-auction} provides a stylized illustration of a pod auction instance. The pod consists of a sequence of ad slots (positions). It has constraints on the maximum number of ads in the pod, and the maximum total duration of the ad break.
There are eight candidate ads bidding into the pod auction, with varying ad durations. Typically, ads have a single value for being shown, uniform across the positions. It is also possible for ads to place per-position bids, like the ad in the bottom-left in the figure; the first and last position in the ad break are often considered the most valuable.\footnote{See section 7.6 of the OpenRTB 2.6 protocol for a specification of pod bidding; \url{https://iabtechlab.com/wp-content/uploads/2022/04/OpenRTB-2-6_FINAL.pdf.}} Another common feature of pod auctions is that advertisers may specify exclusion constraints to avoid being shown along with competitors in the same pod (e.g., two different brands of running shoes may not want to show together); in the figure this is shown as a red line between two of the ads. The auction computes the pod allocation that maximizes the total bid value, subject to the contraints, and charges each winning advertiser its bid. The problem of computing an optimal pod is quite general: it includes as special cases the assignment problem~\citep{koopmans1957assignment}, the knapsack problem, and even independent set.

The generalization of minimum-bid-to-win to pod auctions raises several questions. How should the feedback be encoded? Should guidance be given on how to win each position in the pod? In a video pod auction, there can be several winning ads shown. Should the feedback guarantee that they remain winning if they all updates their bids at once? The same question holds for losing bidders, but there it is even more delicate because not all combinations of ads are feasible. Finally, if several different generalizations are proposed, how should they be evaluated against each other?

To appreciate that the question of minimum-bid-to-win is not straightforward in pod auctions, consider the following example. There are three ads competing for placement in a pod with two positions and a maximum duration of 30s. The first ad has a 30s duration and bids 10 to be shown, while the second and third ads both have a 15s duration and also bid 10 to be shown. The optimal pod contains the latter two ads. Holding the other bids fixed, the second ad could bid as low as 0 (or some small increment above 0) and still be shown, and the same argument holds for the third ad. However, the publisher might find it peculiar to send feedback to both the second and third ads saying they can just bid 0; after all, there is real competition for placement in the pod on the part of the first ad.

This pod auction example is equivalent to a classic example in the literature on combinatorial auctions showing that the VCG mechanism can lead to zero revenue~\citep{ausubel2006lovely,day2008core}. As the example demonstrates, one problem with using VCG payments as feedback is that win guarantees no longer hold if several agents adopt their feedback at once. Another problem is that VCG payments are only defined for winning bidders. For losing bidders they are 0, which is not informative and not a generalization of the minimum-bid-to-win for a single ad slot.
\begin{figure}[t!]
\centering
\includegraphics[scale=0.25]{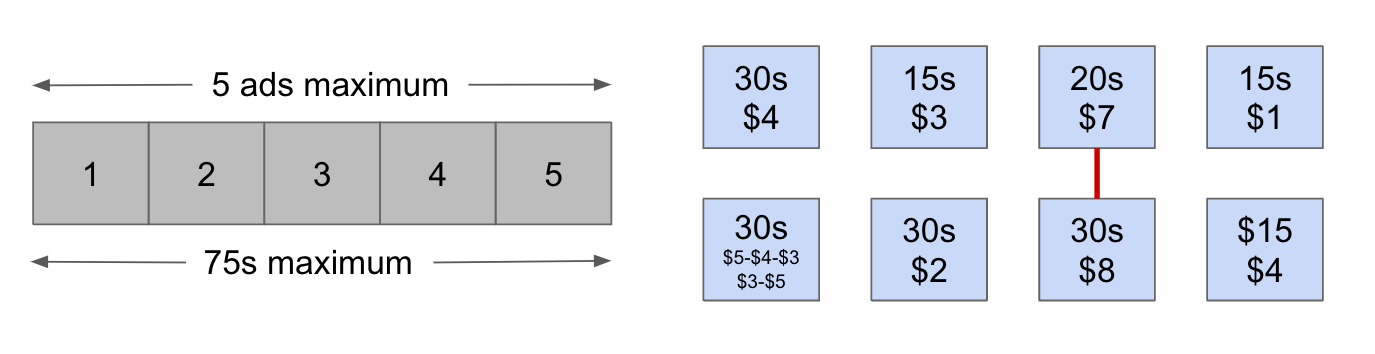}
\vspace{-10pt}
\caption{Illustration of a pod auction instance.}
\label{fig:video-pod-auction}
\end{figure}

\paragraph{Our contributions.} In this paper we propose and evaluate three candidates for a generalized minimum-bid-to-win policy, building on the concepts of VCG payments and the core from cooperative game theory. We first formalize the question as one of recommending uniform discounts or raises to the agents' bids. The feedback is always a single scalar to each bidder, which avoids any difficulties with sending feedback in some bidding language. We justify our setup via established results on equilibria in profit-target strategies in first-price combinatorial auctions~\citep{bernheim1986menu,milgrom2004putting}. 

We then characterize the set of bid updates (discounts or raises) which maintain the optimality of the current allocation, first assuming that just one bidder updates at a time, and then covering the case where several winning bidders (or several losing bidders) update their bids at once. Our main result provides a complete characterization for updates by any mix of winning and losing bidders, using the concept of a bicooperative game (an extension of cooperative games)~\citep{bilbao2000cooperative}. Under this framework, a generalized minimum-bid-to-win policy is any policy that selects a maximal vector of discounts and raises from within these characterizations.

We first consider updates by a single bidder at a time. Here the VCG payment provides feedback to winning bidders, and we define a symmetric concept called the `VCG raise' to provide feedback to losing bidders, emphasizing the parallels between the two concepts. We next turn to updates by several winning bidders at once, or several losing bidders at once. We show that the core of the allocation problem (in the sense of cooperative game theory) characterizes the feedback for these bidders. We show that the same holds for losing bidders with respect to a different coalitional value function, the same function which is used to define VCG raises. To prove our main result on updates from general combinations of bidders, we propose a new (to our knowledge) definition for the core of a bicooperative game, which we call the \emph{bicore}. Our concept bears similarity with the core concept proposed by~\citet{bilbao2007core}, with some key differences. For the assignment problem---a special case of video pod auctions---we give a linear programming characterization of the bicore, extending classic results by~\citet{shapley1971assignment}.

Our results lead to three candidates for a generalized minimum-bid-to-win policy, based on VCG payments, the core, and the bicore. Taking inspiration from the work of~\citet{hoy2013dynamic}, we propose to evaluate the candidates based on the bidding dynamics they induce. 
We simulate the dynamics of the feedback policies on the video pod auction logs of a real ad exchange, and report on how the three compare in terms of their convergence rate and the efficiency of their final allocation.

\paragraph{Incentives.} We must emphasize that although our work draws on concepts like VCG payments and the core, it is not directly concerned with incentives. The agents in our model are not strategic. They have a private value, but for the bulk of the paper values do not matter. The central question is how to provide informative feedback to agents purely based on their current bids. The values come into play when we define bidding dynamics to compare the different feedback policies in simulations, and to evaluate the efficiency of their endpoints. In these dynamics agents never raise their bids above their values because they are participating in first-price auctions.

\paragraph{Related Work.} Our work connects with many lines of research in computer science, operations research, and economics. Bidder feedback and support for combinatorial auctions has been an active area of research. \citet{adomavicius2005toward} introduce the concepts of winning and deadness levels to communicate the lowest level at which a bid could win or have a chance to win in the future, in a model where bids for packages persist across rounds. The evaluation of these feedback mechanisms ranges from complexity analyses to lab experiments~\citep{adomavicius2012effect,petrakis2013ascending}.

The concept of the core features prominently in the literature on combinatorial auctions~\citep{ausubel2002ascending,parkes2000iterative}. Core-selecting payments address several issues inherent in the VCG mechanism, such as low revenue and collusive behavior~\citep{ausubel2006lovely,day2007fair,day2008core}, and there is ongoing research on algorithms to compute core payments~\citep{day2012quadratic,bunz2015faster,bunz2022designing}. Our framing of the problem of providing bidder feedback was strongly influenced by this line of work.

Our results on bidder feedback in the assignment problem build on classic results by~\citet{shapley1971assignment} and~\citet{leonard1983elicitation}. Linear programming characterizations of the core exist for many other combinatorial optimization problems, including facility location, spanning trees, and routing~\citep{goemans2004cooperative,granot1984core,owen1975core}. To prove our main result, we found it necessary to define a value function over a pair of coalitions. This led us to the concept of a bicooperative game, introduced by~\citet{bilbao2000cooperative}. Many fundamental concepts from cooperative game theory have extensions to bicooperative games, including the Shapley value, the Weber set, and most notably the core~\citep{bilbao2007core}.

The problem of bidding in first-price auctions has seen renewed interest since display advertising's shift to this format~\citep{gligorijevic2020bid,zhou2021efficient,paes2020competitive,conitzer2022pacing}. Of course, first-price auctions have long been studied in economics. \citet{bernheim1986menu}'s results on first-price menu auctions were particularly important for our work. We adopt their model of a menu auction in order to cover a wide range of allocation problems, including video pods and combinatorial auctions. 
Their results on equilibria in profit-target strategies provide a justification for our focus on reporting discounts or raises as feedback. 
For our experimental evaluation, we took inspiration from the work of~\citet{hoy2013dynamic} who analyze bidding dynamics in first-price auctions with profit-target bidding. Among many relevant results, they make the important observation that equilibrium can be restored in a first-price auction by augmenting the bidding language with a single scalar profit-target. 

In the domain of video ads,~\citet{nisan2009google} report on an early Google system (now discontinued) where advertisers could bid on TV ad slots one day ahead. Today's video pod auctions run in real-time when the content starts or even right before the ad break (e.g., in live-streamed events).~\citet{goetzendorff2015compact} introduce a bidding language for TV advertising and investigate algorithms for winner determination and core payments. The pods in their allocation problem have essentially the same structure and constraints as online video pods, but in their model several pods are auctioned off at once ahead of time. 

\paragraph{Outline.}

The remainder of the paper is organized as follows. In Section~\ref{sec:prelims} we introduce and justify our model, and give an initial formalization of the concept of minimum-bid-to-win. In Section~\ref{sec:single-vcg} we consider scenarios where a single agent updates its bid at a time, and illustrate the VCG discount and VCG raise via two examples. In Section~\ref{sec:multi-core} we consider scenarios where multiple winning bidders (or multiple losing bidders) can update their bids at once, making a connection to the core. We then prove our main result on combinations of updates by any mix of winning and losing bidders. In Section~\ref{sec:assignemnt-problem} we give a linear programming characterization of the bicore of the assignment problem and prove that it has a lattice structure. 
Proofs from this section are deferred to the appendix.
In Section~\ref{sec:bidding-dynamics} we introduce bidding dynamics, and in Section~\ref{sec:empirical-analysis} we report on results from simulations of the dynamics under feedback policies based on VCG, the core, and the bicore. Section~\ref{sec:conclusion} concludes.


\section{Model and Preliminaries}
\label{sec:prelims}

We consider a model where a seller (auctioneer) selects an allocation from a finite set $\aX$ based on bids placed by a set of agents (bidders). Let $N$ denote the set of agents and let $n = |N|$. Each agent $i$ has a non-negative \emph{valuation} function $v_i: \aX \rightarrow \reals_+$ mapping allocations to non-negative values, drawn from a set $\aV_i$. Agents have quasi-linear utility: when the selected allocation is $x \in \aX$ and agent $i$ is charged $p_i$, its utility is $v_i(x) - p_i$. The \emph{support} of agent $i$'s valuation is defined as $\dom v_i = \{x \in \aX : v_i(x) > 0\}$. The seller itself has zero reservation value for any allocation; its utility is the total revenue $\sum_{i \in N} p_i$ collected from the agents.

Each agent $i$ participates in an auction run by the seller by submitting a \emph{bid} function $b_i : \aX \rightarrow \reals_+$, drawn from a set $\aB_i$. Naturally, the bid may be different from the agent's actual valuation. We say that an allocation $x$ is \emph{optimal} if it maximizes $\sum_i b_i(x)$, whereas we say it is \emph{efficient} if it maximizes $\sum_i v_i(x)$.
Given the agents' bids $\vb = (b_1, \dots, b_n)$, the set of optimal allocations is denoted as $\aX^* = \argmax_{x \in \aX} \sum_{i \in N} b_i(x)$. We will write $\aX^*(\vb)$ whenever we find it necessary to make the bids defining the optimal allocations explicit.
The seller runs a \emph{first-price} (or \emph{pay-your-bid}) auction, which proceeds as follows:
\begin{enumerate}
\item Each agent $i$ submits a bid $b_i \in \aB_i$.
\item The seller selects an optimal allocation $x^* \in \aX^*(\vb)$.
\item The seller charges each $i$ the amount $p_i = b_i(x^*)$.
\end{enumerate}

\paragraph{Bid updating.}
In this work we are interested in understanding whether a bid update would change the optimal allocation, and how the update might affect whether an agent "wins" or "loses". We therefore require some structure to the $\aB_i$ bid spaces. Each $\aB_i$ is agent-specific because it depends on the agent's valuation $v_i$. First, we require that $v_i \in \aB_i$, and that $b_i \leq v_i$ for all $b_i \in \aB_i$, where the inequality is understood point-wise. (This condition is justified by the fact that bidding higher than one's value is a dominated strategy in a first-price auction.) In particular, $b_i(x) = 0$ for $x \not\in \dom v_i$; however, it is possible to have $b_i(x) = 0$ even if $x \in \dom v_i$.
Next, we require that the set of bids $\aB_i$ be closed under two kinds of updates: bid discounting and bid raising.
\begin{definition} \label{def:bid-updates}
Given a bid function $b_i$ and a scalar \emph{discount} $\dis_i \geq 0$, the \emph{discounted bid} function $\dbid_i$ is defined as $\dbid_i(x) = \max\{b_i(x) - \pi_i, 0\}$. Given a scalar \emph{raise} $\rai_i \geq 0$, a \emph{raised bid} function $\rbid_i$ is defined as any function that satisfies $\rbid_i(x) = b_i(x) + \mu_i$ for $b_i(x) > 0$ and $\rbid_i(x) \leq \mu_i$ for $b_i(x) = 0$.
\end{definition}
\noindent
Note that a bid maps to a unique discounted bid, but there may be a set of possible raised bids for any original bid. This reflects the fact that bid discounting can lose information when the bid on a particular allocation reaches 0. This means that there is no single reverse operation, so for an allocation that has a bid of 0 an agent can raise its bid however it wishes up to $\rai_i$. We emphasize one exception to this: since $b_i \leq v_i$ for all $b_i \in \aB_i$, and $\aB_i$ is closed under bid-raising, a raised bid must satisfy $\rbid_i(x) = 0$ for all $x \not\in \dom v_i$. In other words, an agent's bid on an specific allocation can only rise above 0 if the agent's value for the allocation is positive.

\paragraph{Winning and losing.}
The main reason for introducing the concept of the support of a valuation $v_i$ is that it allows us to formally define whether an agent is "winning" or "losing". Intuitively, an agent wins if it has positive value for the selected allocation.
\begin{definition}
 Given an allocation $x \in \aX$, we say that an agent \emph{wins} in the allocation if $x \in \dom v_i$, and that it \emph{loses} otherwise. We say that an agent is \emph{winning} (\emph{losing}) if there exists an optimal allocation $x^* \in \aX^*$ in which it wins (loses).  
\end{definition}
\noindent
Note that according to these definitions an agent can be both winning and losing. Although this may seem contradictory at first glance, it simply reflects the fact that tie-breaking may ultimately determine whether the agent wins in the selected optimal allocation.

By working with an abstract set of feasible allocations $\aX$ our intention is to capture, in a generic way, the variety of allocation constraints that migth arise in specific applications. For instance, in a video pod auction, $\aX$ would disallow any allocation that shows two ads that have an exclusion constraint between them. A key requirement for our results is that the sets of allocations and bids (i.e., the environment) satisfy the following property.
\begin{property}
An environment exhibits \emph{agent-free-disposal} if for all $b_i \in \aB_i$ ($i \in N$), $S \subseteq N$, and $x \in \aX$ there exists $x' \in \aX$ such that $\sum_{i \not\in S} b_i(x') \geq \sum_{i \not\in S} b_i(x)$ and $x' \not\in \dom v_i$ for $i \in S$.
\end{property}
\noindent
The property intuitively captures the option of dropping a set of agents from an allocation, at no collective detriment to the other agents. We will assume that agent-free-disposal holds throughout the paper.

\paragraph{Minimum-bid-to-win.}
We are now in a position to generalize the concept of minimum-bid-to-win from a single-item auction to the general auctions over allocations that we have defined so far. Given bids $\vb = (b_1, \dots, b_n)$, we form the following set of discounts for a \emph{winning} agent $i$:
\begin{equation} \label{eq:mbtwin-set-i}
\mbtw_i = \left\{ \dis_i \geq 0 : \aX^*(\vb) \subseteq \aX^*(\dbid_i \,, \vb_{-i}) \right\}.  
\end{equation}
Take the maximum discount from $\mbtw_i$.
The \emph{minimum-bid-to-win} for $i$ is the discounted bid $\dbid_i$ which results from applying this discount. 
Symmetrically, we form the following set of raises for a \emph{losing} agent $i$:
\begin{equation} \label{eq:mbtlose-set-i}
\mbtl_i = \left\{ \rai_i \geq 0 : \aX^*(\vb) \subseteq \aX^*(\rbid_i \,, \vb_{-i}) \right\}.
\end{equation}
The \emph{maximum-bid-to-lose} is the raised bid that results from the maximum element of this set.

We do not introduce any special notation for the minimum-bid-to-win or maximum-bid-to-lose, because our focus in the rest of the paper will be on characterizing sets of discounts and raises like $\mbtw_i$ and $\mbtl_i$ under which optimal allocations remain optimal. Given such a characterization, a minimum-bid-to-win (maximum-bid-to-lose) is any discounted (raised) bid that results from a \emph{maximal} element of the set. The advantage of this perspective is that discounts and raises are simple scalars, whereas bid functions are more complicated objects whose structure is application-dependent. For instance, for computational reasons, the auction might be conducted using a bidding language~\citep{nisan2000bidding,boutilier2001bidding}.

\paragraph{Profit-target strategies.}

Our focus on bid updates of the specific form given in Definition~\ref{def:bid-updates} calls for some justification, as many other kinds of updates are possible. For this we appeal to~\citet{bernheim1986menu}'s results on menu auctions, which apply to our model; we will refer to the terminology and coverage of these results given by~\citet[Chapter 8]{milgrom2004putting}. A bid function $b_i$ is a \emph{$\rho_i$-profit-target strategy} if it takes the form $b_i(x) = \max\{v_i(x) - \rho_i, 0\}$. In other words, it is the bid function that results from applying a discount of $\rho_i$ to the agent's valuation, and by our assumptions this bid function is contained in $\aB_i$ for any profit-target $\rho_i$.

\citet{bernheim1986menu} show that in a first-price menu auction, no matter what the other agents' bids, an agent always has a best-reply in the form of a profit-target strategy~\cite[Theorem 8.6]{milgrom2004putting}. Furthermore, the vector of profit-targets $(\rho_1, \dots, \rho_n)$ is a bidder-optimal core payoff if and only if the corresponding profit-target strategies form a Nash equilibrium~\cite[Theorem 8.7]{milgrom2004putting}.\footnote{See Section~\ref{sec:multi-core} for background on the core and bidder-optimal core payoffs.} In particular, this means that Nash equilibria in such strategies are efficient. Taken together, these results support using discounts and raises as informative feedback for bidders: they can be interpreted as recommendations as to how agents should update their profit-targets. A discount (raise) is a recommendation for a bidder to increase (decrease) its profit-target by a specific amount.

In practice, the auction's bidding language may not be expressive enough on its own to represent all possible profit-target strategies. As explained in the introduction, current minimum-bid-to-win feedback in first-price display ad auctions consists of a single scalar in the message sent to a bidder. However, this is a minor obstacle since profit-target strategies can always be represented by augmenting the bidding language by a single scalar~\citep{hoy2013dynamic}. In display ad auctions, an extra field could be added to the feedback message, or the minimum-bid-to-win field itself could be repurposed to encode a discount or raise.


\section{Single-Agent Updates and VCG}
\label{sec:single-vcg}

In this section we study minimum-bid-to-win concepts based on bid updates from a single agent at a time, holding the other agents' bids fixed. This connects with concepts from the celebrated \emph{Vickrey-Clarke-Groves mechanism} (or VCG mechanism)~\citep{vickrey1961counterspeculation,clarke1971multipart,groves1973incentives}. However, the theory of VCG mechanisms only leads to feedback for winning bidders. Here we develop analogous concepts that provide feedback for losing bidders as well.

\subsection{VCG Discounts and Raises}

The well-known \emph{VCG payment} is a natural generalization of the concept of minimum-bid-to-win for a winning agent. In keeping with common notation used to describe VCG payments, for an agent $i$ we write $N - i$ rather than $N \setminus \{i\}$. Consider the following function over subsets of bidders $S \subseteq N$, also called \emph{coalitions}:
\begin{equation} \label{eq:winning-coalitional-value}
    \Vw(S) = \max_{x \in \aX} \sum_{i \in S} b_i(x).
\end{equation}
Here the subscript $w$ alludes to the fact that the function is associated with winning bidders. The \emph{VCG discount} for a winning agent $i$ is defined as $\Vw(N) - \Vw(N-i)$, and its \emph{VCG payment} for the selected allocation $x^* \in \aX^*$ is $b_i(x^*) - [ \Vw(N) - \Vw(N-i) ]$. In other words, the VCG payment is the agent's bid for $x^*$ under the discounted bid $\dbid_i$ where it has applied the discount $\pi_i = \Vw(N) - \Vw(N-i)$. (The fact that the VCG discount is no greater than $b_i(x^*)$ will be confirmed shortly.)
It is a standard result that the VCG payment corresponds to the minimum bid that an agent can place and still remain winning, holding the other agents' bids fixed~\cite[Chapter 2]{milgrom2004putting}. 

The VCG discount is therefore a meaningful discount to communicate as feedback for winning bidders under the framework established in Section~\ref{sec:prelims}. However, for losing bidders we have $\Vw(N) = \Vw(N-i)$ so the VCG discount is 0, which is uninformative. We propose here a symmetric concept that we call the `VCG raise'. Consider the following coalitional value function:
\begin{equation} \label{eq:losing-coalitional-value}
    \Vl(S) = \max_{x \in \aX} \left\{ \sum_{i \in N} b_i(x) \,:\, x \in \dom v_i,\, \forall i \in N \setminus S \right\}.
\end{equation}
The subscript $\ell$ alludes to the fact that the function is associated with losing bidders. To understand why, note that for a losing bidder $i$, $\Vl(N - i)$ corresponds to the optimal bid value among allocations where $i$ is constrained to win. We define the \emph{VCG raise} as $\Vl(N) - \Vl(N-i)$. Our claim is that it is the largest raise that losing agent $i$ can apply to its bid and still remain losing, holding the other agents' bids fixed. The following proposition follows as a corollary to our main theorem in Section~\ref{sec:multi-core}.
\begin{proposition} \label{prop:vcg-characterizations}
The sets $\mbtw_i$ and $\mbtl_i$ defined in equations~\eqref{eq:mbtwin-set-i}--\eqref{eq:mbtlose-set-i} are characterized by
\begin{eqnarray*}
\mbtw_i & = & \left\{ \dis_i : 0 \leq \dis_i \leq \Vw(N) - \Vw(N-i) \right\} \\   
\mbtl_i & = & \left\{ \rai_i : 0 \leq \rai_i \leq \Vl(N) - \Vl(N-i) \right\}    
\end{eqnarray*}
\end{proposition}
It is possible for the maximization problem defining $\Vl$ in~\eqref{eq:losing-coalitional-value} to be infeasible for a particular set $S$. In that case, we set it to $\Vw(S) = -\infty$ following the usual convention for infeasible maximizations. If $\dom v_i = \emptyset$ (i.e., the agent has zero value for every allocation), then $\Vw(N-i) = -\infty$, and according to the previous proposition $\mbtl_i$ consists of all raises $\mu_i \geq 0$, taking $-(-\infty) = +\infty$. This reflects the fact that the agent remains losing no matter how much it raises its bid. However, in practice such an agent would have no reason to participate in the auction, and could be excluded prior to running the auction. In the remainder of the paper, we will assume that for each $i \in N$, $\dom v_i \neq \emptyset$, ensuring that $\mbtw_i$ and $\mbtl_i$ are bounded for all agents.

\begin{remark}
There is conceptual symmetry between VCG discounts and VCG raises. Note that for a winning bidder $i$, we have $\Vl(N) - \Vl(N - i) = 0$, so its VCG raise is 0, in analogy with the VCG discount being 0 for losing bidders. If $x^*$ is an allocation maximizing~\eqref{eq:winning-coalitional-value} for coalition $N$ (i.e., if it is an optimal allocation), then for a winning bidder $i$ the discounted bid $\dbid_i(x^*) = b_i(x^*) - [\Vw(N)- \Vw(N-i)]$ (i.e., the VCG payment) does not depend on its original bid $b_i(x^*)$, as the latter cancels out. 
Similarly, for a losing bidder $i$, if $x^*$ is an allocation maximizing~\eqref{eq:losing-coalitional-value} for coalition $N-i$ then the raised bid $\dbid_i(x^*) = b_i(x^*) + [\Vl(N)- \Vl(N-i)]$ does not depend on its original bid $b_i(x^*)$, as again the latter cancels out. These facts reflect the intuition that the minimum-bid-to-win or maximum-bid-to-lose for an agent should not depend on the agent's own bid, only the other agents' bids. There is also symmetry in the definitions of $\Vw$ and $\Vl$. By the agent-free-disposal property, the maximization in~\eqref{eq:winning-coalitional-value} can be reformulated as maximizing the total bid value, subject to the constraint that agents in $N \setminus S$ are losing. This is directly analogous to~\eqref{eq:losing-coalitional-value}, which maximizes the total bid value subject to the constraint that agents in $N \setminus S$ are winning.
\end{remark}

\subsection{Two Examples}
\label{ex:vcg-examples}

Here we record two examples which we will return to later to understand the similarities and differences between alternative feedback policies. The first is the single-item auction; although this example is basic, the feedback policies are not all identical even in this special case. The second is the three-agent video pod auction example from the introduction, equivalent to a classic example given by~\citet{ausubel2006lovely} to show that the VCG mechanism can lead to zero revenue despite competition for the items up for sale.

\begin{example} {\em (Single-item auction)}\label{ex:single-item-auction} 
\ Suppose there is a single indivisible item and that the bids are $b_1 > b_2 > \dots > b_n$. The VCG discount for winning agent 1 is $\Vw(N) - \Vw(N-1) = b_1 - b_2$, leading to the discounted bid $\dbid_1 = b_1 - (b_1 - b_2) = b_2$, as expected. For each losing agent $i \neq 1$, the VCG raise is $\Vl(N) - \Vl(N-i) = b_1 - b_i$, leading to the raised bid $\rbid_i = b_i + (b_1 - b_i) = b_1$. Using VCG feedback therefore reproduces the current minimum-bid-to-win feedback policy for single slots in display advertising.
\end{example}
\begin{example} {\em (Zero-VCG)}\label{ex:zero-vcg} 
\ Consider a video pod auction with two positions and maximum duration of 30s. There are three agents with uniform values across positions. Agents are defined by value-duration pairs: (10, 30s) for agent 1, (10, 15s) for agent 2, (10, 15s) for agent 3. Suppose the agents bid their values. The optimal pod contains the ads from agents 2 and 3. For losing agent 1, the VCG raise is 10. However, the agent would not apply this raise since it is already bidding its value. For winning agents 2 and 3, the VCG discount in each case is 10. This means that the discounted bid for each agent to stay winning is 0.
\end{example}

\subsection{Valid Discounts and Raises}

Our main characterization result will rely on certain basic conditions that bidder feedback should satisfy.
The following definition captures these conditions.
\begin{definition}
Given bids $(b_1, \dots, b_n)$, the discounts $\vdis = (\dis_1, \dots, \dis_n)$ and raises $\vrai = (\rai_1, \dots, \rai_n)$ are \emph{valid} if the following conditions hold.
\begin{itemize}
    \item For every winning agent $i$, $\rai_i = 0$ and $0 \leq \dis_i \leq \min_{x^* \in \aX^*} b_i(x^*)$.
    \item For every losing agent $i$, $\dis_i = 0$ and $\rai_i \geq 0$.
\end{itemize}
\end{definition}
\noindent
To interpret this definition, first note that a winning agent should not want to raise its bid (so $\mu_i = 0$), and similarly a losing agent should not want to discount its bid (so $\pi_i = 0$). Moreover, if the selected optimal allocation is $x^* \in \aX^*$, then a winning agent $i$'s discounted bid for this allocation reaches 0 once $\pi_i = b_i(x^*)$. Larger discounts are not meaningful because they have no further effect on the discounted bid for $x^*$, so it is natural to impose $\dis_i \leq b_i(x^*)$. To avoid dependence on any particular tie-breaking scheme, the definition imposes this condition for all $x^* \in \aX^*$.

\begin{lemma} \label{lem:core-valid}
If $\dis_i \in \mbtw_i$ and $\rai_i \in \mbtl_i$ for all $i \in N$, then $\vdis$ and $\vrai$ are valid discounts and raises.     
\end{lemma}
\begin{proof}
By the definition of $\mbtw_i$ and $\mbtl_i$, $\vdis$ and $\vrai$ are non-negative. If $i$ is a winning agent, then $\Vl(N) = \Vl(N - i)$, which implies that $\rai_i = 0$. If $i$ is a losing agent, then $\Vw(N) = \Vw(N - i)$, which implies that $\dis_i = 0$.

Let $x^* \in \aX^*$. For any $j \in N$, we have $\Vw(N) = \sum_{i} b_i(x^*)$ and $\Vw(N - j) \geq \sum_{i \neq j} b_i(x^*)$, from which we obtain $\Vw(N) - \Vw(N-j) \leq b_j(x^*)$. Combined with the characterization of $\mbtw_j$ in Proposition~\ref{prop:vcg-characterizations}, this implies $\pi_j \leq b_j(x^*)$. As $x^*$ was an arbitrary optimal allocation, this completes the proof.
\end{proof}
%


\section{Multi-Agent Updates and the Core}
\label{sec:multi-core}

If multiple winning agents follow the feedback specified by VCG discounts simultaneously then the agents are no longer guaranteed to win, as we saw in Example~\ref{ex:zero-vcg}. In fact, in the example all the winning agents would decrease their bids to 0. In this section we develop feedback that provides guarantees for coalitions of agents, drawing on concepts from cooperative game theory.

\subsection{The Core}

We first examine feedback for winning bidders. A \emph{cooperative game} is defined by a set of players together with a value function defined over coalitions~\citep{myerson1997game}. In our setting the set of players is $N$ together with the seller, for which we use the special index $s$, and the value function is $\Vw$ defined in~\eqref{eq:winning-coalitional-value}. The \emph{core} of this cooperative game is a solution concept that distributes the optimal allocation value among the players. For an auction with a single seller, the core is the set of all non-negative payoff vectors $(\dis^s, \vdis)$ satisfying the following conditions~\citep{bikhchandani2002package}.
\begin{eqnarray}
\dis^s + \sum_{i \in N} \dis_i & = & \Vw(N) \label{eq:core-eq} \\
\dis^s + \sum_{i \in S} \dis_i & \geq & \Vw(S) \hspace{2em} \forall S \subseteq N \label{eq:core-lb}
\end{eqnarray}
Here $\dis_i$ is normally interpreted as the utility (payoff) to bidder $i$, while $\dis^s$ is the revenue to the seller. Condition~\eqref{eq:core-eq} states that the total utility to all auction participants (including the seller) should amount to the optimal allocation value. Condition~\eqref{eq:core-lb} states that the total utility to participants in $S$ and the seller should exceed the allocation value $\Vw(S)$ that these players can together realize for themselves; this is a `stability' condition to ensure that no coalition would want to defect. For coalitions that do not contain the seller, or for the seller on its own, the optimal value is simply zero, which leads to the condition that $(\dis^s, \vdis)$ should be non-negative.

The inequalities defining the core can be simplified so that the seller term $\dis^s$ drops out. We define the set
\begin{equation} \label{eq:core-ub}   
\core(N, \Vw) = \left\{ \vdis \geq 0 : \sum_{i \in S} \dis_i \leq \Vw(N) - \Vw(N \setminus S),\, \forall S \subseteq N \right\}.
\end{equation}
It is straightforward to show that a vector $\vdis \in \core(N, \Vw)$, together with $\dis^s = \Vw(N) - \sum_i \dis_i$, is in the core, and vice-versa:~\eqref{eq:core-eq} and~\eqref{eq:core-lb} imply the inequalities in~\eqref{eq:core-ub}, and the argument can be reversed.

Note that~\eqref{eq:core-ub} contains the inequalities defining $\mbtw_i$ for all $i \in N$, taking $S = \{i\}$. Therefore we have $\core(N, \Vw) \subseteq \left( \times_i \mbtw_i \right)$, and by Lemma~\ref{lem:core-valid} a vector $\vdis \in \core(N, \Vw)$ represents valid discounts. Let $\vzero$ be the $n$-dimensional zero vector. It is also the case that $(\pi_i, \vzero_{-i}) \in \core(N, \Vw)$ for $\pi_i \in \mbtw_i$. These facts correspond to the known results that VCG discounts upper-bound core payoffs, and that each agent's VCG discount is attained at some core vector~\citep{bikhchandani2002package}.

Developing analogous concepts relevant to losing bidders is now immediate: take the core $\core(N, \Vl)$ with respect to the coalitional value function~\eqref{eq:losing-coalitional-value}. All the results mentioned in the previous paragraph hold, replacing the discounts $\vdis$ with raises $\vrai$ and $\mbtw_i$ with $\mbtl_i$. (Naturally, for a VCG raise to be attained at some vector in $\core(N, \Vl)$, the raise has to be finite.) The following result follows as a corollary to our main result in the next section.
\begin{theorem} \label{thm:core-results}
We have $\vdis \in \core(N, \Vw)$ if and only if $\vdis$ are valid discounts and, for all $S \subseteq N$, every optimal allocation under the original bids remains optimal after every bidder $i \in S$ discounts its bid by $\dis_i$.
Analogously,
we have $\vrai \in \core(N, \Vl)$ if and only if $\vrai$ are valid raises and, for all $T \subseteq N$, every optimal allocation under the original bids remains optimal after every bidder $i \in T$ raises its bid by $\rai_i$.
\end{theorem}
To gain some intuition for one direction of this result, let $x^*$ be an optimal allocation and let $T$ be a coalition of bidders that lose in $x^*$. Let $x \in \aX$ be an allocation that maximizes the bid value subject to the constraint that all $i \in T$ are winning. The bid value of $x$ is $\Vl(N \setminus T)$ by definition. If each $i \in T$ raises its bid by $\rai_i$, the value of $x$ increases to $\Vl(N \setminus T) + \sum_{i \in T} \rai_i$, while the value of $x^*$ remains at $\Vl(N)$. To ensure that $x^*$ remains optimal, we must have $\Vl(N \setminus T) + \sum_{i \in T} \rai_i \leq \Vl(N)$, which is one of the defining constraints for $\core(N, \Vl)$.

\subsection{The Bicore}

The previous section examined joint bid updates by coalitions of winning bidders, or coalitions of losing bidders, but not the most general case of bid updates by winning and losing bidders in tandem, which we now cover here. To state our main characterization we result make use of the concept of a \emph{bicooperative game}, introduced by~\citet{bilbao2000cooperative}. In a bicooperative game, the value function is defined over pairs of coalitions $(S, T)$ where $S,T \subseteq N$ and $S \cap T = \emptyset$. The interpretation is that $S$ is a coalition of `cooperators' whereas $T$ is a coalition of `detractors'.

We define a bicooperative game where the players are again the agents $N$ together with the seller, and the coalitional value function is defined as follows:
\begin{equation} \label{eq:bicoalitional-value}
V(S,T) = \max_{x \in \aX} \left\{\sum_{i \in S \cup T} b_i(x) : x \in \dom v_i,\: \forall i \in T \right\}.
\end{equation}
In words, $V(S,T)$ represents the maximum bid value that can be collectively achieved for agents $S \cup T$, subject to the constraint that all agents in $T$ are winning. If there is no allocation that satisfies this constraint, then we set $V(S,T) = -\infty$.

For bicooperative games based on our model where there is a single seller, we propose the following concept of the core, which we call the `bicore'.

\medskip
\begin{definition} \label{def:bi-core}
The \emph{bicore} is the set of all non-negative payoff vectors $(\dis^s, \vdis, \vrai)$, where $\vdis = (\dis_1, \dots, \dis_n)$ and $\vrai = (\rai_1, \dots, \rai_n)$, satisfying the following conditions.
\begin{eqnarray}
\dis^s + \sum_{i \in N} \dis_i & = & V(N, \emptyset) \label{eq:bicore-eq} \\
\dis^s + \sum_{i \in S} \dis_i - \sum_{j \in T} \rai_j & \geq & V(S, T) \hspace{2em} \forall S,T \subseteq N \label{eq:bicore-lb}
\end{eqnarray}
\end{definition}
Here the $\pi_i$ have their usual interpretations as bidders utilities, while $\pi^s$ corresponds to the seller's revenue.
%
%
The bicore includes a new term $\mu_i$ for each $i \in N$ which can be interpreted as the \emph{disutility} of having to include $i$ in the winning coalition. Condition~\eqref{eq:bicore-lb} then states that the total utility to bidders in $S$ (along with the seller), net of the disutility from the constraint that bidders in $T$ must win, should still exceed the constrained coalitional value $V(S,T)$.

As with the core, the seller term $\pi^s$ can be made implicit, leading to a more informative formulation for our purposes. A pair of discount and raise vectors $(\vdis, \vrai)$ is in the bicore if and only if it lies in the following set:
\begin{equation} \label{eq:bicore-ub}
\bicore(N, V) = \left\{ (\vdis, \vrai) \geq 0 : \sum_{i \in S} \dis_i + \sum_{j \in T} \rai_j \leq V(N, \emptyset) - V(N \setminus S, T),\, \forall S,T \subseteq N \right\}.
\end{equation}
Note that according to~\eqref{eq:bicoalitional-value} we have $V(S,T) = V(S \setminus T, T)$. It therefore suffices to specify $V(S,T)$ only for coalitions that satisfy $S \cap T = \emptyset$, in agreement with the definition of a bicooperative game given by~\citet{bilbao2000cooperative}. It is also straightforward to check that any of the constraints~\eqref{eq:bicore-lb} beyond those for which $S \cap T = \emptyset$ are redundant. Nonetheless, in proving our results, we have sometimes found it more straightforward prove that~\eqref{eq:bicore-lb} or~\eqref{eq:bicore-ub} hold for all $S,T \subseteq N$.
\begin{remark}
Our Definition~\ref{def:bi-core} for the bicore bears strong resemblance to Definition 2 given by~\citet{bilbao2007core}, but they are not quite equivalent. The definitions share constraints~\eqref{eq:bicore-ub}, but instead of constraint~\eqref{eq:bicore-eq} \citet{bilbao2007core} impose the following equality:
\begin{equation} \label{eq:bilbao-preimputation}
\sum_{i \in N} \dis_i + \sum_{i \in N} \rai_i = V(N, \emptyset) - V(\emptyset, N).    
\end{equation}
This corresponds to a different notion of \emph{preimputations}, which are payoff vectors that distribute the optimal allocation value among the players. The left-hand sides of~\eqref{eq:bicore-eq} and~\eqref{eq:bilbao-preimputation} do not match, and in our setting it is not unusual to have $V(\emptyset, N) = -\infty$, which makes~\eqref{eq:bilbao-preimputation} meaningless.
\end{remark}

\subsection{Main Characterization}

We are now ready to prove our main result.
\begin{theorem} \label{thm:main-result}
Let $\vdis$ and $\vrai$ be vectors of discounts and raises, respectively. We have $(\vdis, \vrai) \in \bicore(N, V)$ if and only if $\vdis$ and $\vrai$ are valid discounts and raises, and for all coalitions $S,T \subseteq N$, every optimal allocation under the original bids remains optimal after every bidder $i \in S$ discounts its bid by $\dis_i$ and every bidder $i \in T$ raises its bid by $\rai_i$.
\end{theorem}
\begin{proof}
$(\Leftarrow)$ Let $x^* \in \aX^*$ be an optimal allocation with respect to the original bids. As $\vdis$ and $\vrai$ are valid discounts and raises, they are non-negative. As $\vdis$ are valid discounts, we have $\dis_i \leq b_i(x^*)$ for all $i \in N$, which implies $\sum_i \dis_i \leq \sum_i b_i(x^*) = V(N, \emptyset)$. Therefore, $\dis^s = V(N, \emptyset) - \sum_i \dis_i \geq 0$, and the vector $(\dis^s, \vdis, \vrai)$ has non-negative components. Equality~\eqref{eq:bicore-eq} holds by construction. It remains to establish inequalities~\eqref{eq:bicore-ub}.

Let $S, T \subseteq N$, and assume that every bidder $i \in S$ discounts its bid by $\dis_i$, and every bidder $i \in T$ raises its bid by $\rai_i$. We may assume that $\dis_i > 0$ for each $i \in S$ and $\rai_i > 0$ for each $i \in T$, because for the other bidders the updates do not change their bids. As $\vdis$ and $\vrai$ are valid, this implies that $S \cap T = \emptyset$, and that every $i \in S$ wins in every optimal allocation while every $i \in T$ loses in every optimal allocation. For each $i \in S$ we have $\dbid_i(x^*) = b_i(x^*) - \pi_i$ as $\pi_i \leq b_i(x^*)$, because $\pi_i$ is a valid discount. For each $i \in T$ we have $\rbid_i(x^*) = b_i(x^*) = 0$ as $x^* \not\in \dom v_i$, because $i$ loses in $x^*$.
The bid value of $x^*$ under the updated bids is therefore:
\begin{equation} \label{eq:updated-Vn}
\sum_{i \in S} \dbid_i(x^*) + \sum_{i \in T} \rbid_i(x^*) + \sum_{i \not\in S \cup T} b_i(x^*) 
= \sum_{i \in N} b_i(x_i^*) - \sum_{i \in S} \pi_i = V(N, \emptyset) - \sum_{i \in S} \pi_i.
\end{equation}

Now let $x$ be an allocation that satisfies the constraints and achieves the maximum in the optimization for $V(N \setminus S, T)$. By the agent-free-disposal property we can assume that $x \not\in \dom v_i$ for all $i \not\in (N \setminus S) \cup T = N \setminus S$, so that $b_i(x) = 0$ for these agents. Noting that $(N \setminus S) \setminus T = N \setminus (S \cup T)$, and that the bids from these agents do not change, the total bid value of $x$ under the updated bids is therefore:
\begin{equation} \label{eq:updated-Vst}
\sum_{i \in (N \setminus S) \setminus T} b_i(x) + \sum_{j \in T} \rbid_j(x) + \sum_{i \not\in N \setminus S} b_i(x) = \sum_{i \in (N \setminus S) \cup T} b_i(x) + \sum_{i \in T} \rai_i = V(N \setminus S, T) + \sum_{i \in T} \rai_i
\end{equation}
As $x^*$ remains optimal under the updated bids, we have that~\eqref{eq:updated-Vn} is at least~\eqref{eq:updated-Vst}, which proves that~\eqref{eq:bicore-ub} holds. The vector $(\dis^s, \vdis, \vrai)$ is therefore in the bicore.

\medskip
$(\Rightarrow)$ Let $S, T \subseteq N$, and assume that every bidder $i \in S$ discounts its bid by $\dis_i$, and every bidder $i \in T$ raises its bid by $\rai_i$. Again, we may assume that $\dis_i > 0$ for each $i \in S$ and $\rai_i > 0$ for each $i \in T$, because for the other bidders the updates do not change their bids. By Lemma~\ref{lem:core-valid}, $(\vdis, \vrai)$ are valid discounts and raises, so we must have $S \cap T = \emptyset$.

 Let $x^* \in \aX^*$ be an optimal allocation with respect to the original bids. For $i \in T$, as $\rai_i > 0$ and $\vrai$ is valid, $i$ is not winning and $x^* \not\in \dom v_i$. Therefore $\rbid_i(x^*) = b_i(x^*) = 0$. For $i \not\in S \cup T$, the bids of these agents do not change: either $\pi_i = \mu_i = 0$, or $i$ is not an agent that updates its bid. For $i \in S$, as $\vdis$ is valid we have $b_i(x^*) \geq \pi_i$, which implies $\dbid_i(x^*) = b_i(x^*) - \pi_i$. The value of $x^*$ after the bid updates is therefore:
\begin{equation} \label{eq:discounted-opt-value}
    \sum_{i \in S} \dbid_i(x^*) + \sum_{i \in T} \rbid_i(x^*) + \sum_{i \not\in S \cup T} b_i(x^*)
    = \sum_{i \in N} b_i(x^*) - \sum_{i \in S} \pi_i = V(N, \emptyset) - \sum_{i \in S} \pi_i.
\end{equation}
The last equality follows from the fact that $x^* \in \aX^*$.

Now let $x \in \aX$. Let $S^+ = \{i \in S : b_i(x) \geq \pi_i\}$ and let $S^- = S \setminus S^+$. Note that $\dbid_i(x) = 0$ for $i \in S^-$. We have
\begin{equation} \label{eq:discounted-bids}
\sum_{i \in S} \dbid_i(x) = \sum_{i \in S^+} \dbid_i(x) = \sum_{i \in S^+} b_i(x) - \sum_{i \in S^+} \dis_i.    
\end{equation}
Let $T^+ = \{i \in T : x \in \dom v_i\}$ and let $T^- = T \setminus T^+$. We have $\rbid_i(x) \leq b_i(x) + \mu_i$ for $i \in T^+$ and $\rbid_i(x) = b_i(x) = 0$ for $i \in T^-$. Therefore we have
\begin{equation} \label{eq:raise-bids}
\sum_{i \in T} \rbid_i(x) \leq \sum_{i \in T} b_i(x) + \sum_{i \in T^+} \mu_i.    
\end{equation}
From~\eqref{eq:discounted-bids} and~\eqref{eq:raise-bids} we obtain that the total bid value of $x$ after the bid updates satisfies the following (recall that $S \cap T = \emptyset$):
\begin{eqnarray*}
& \ds \sum_{i \in N \setminus (S \cup T)} b_i(x) + \sum_{i \in S} \dbid_i(x) + \sum_{i \in T} \rbid_i(x) & \\
\leq & \ds \sum_{i \in N \setminus (S \cup T)} b_i(x) + \sum_{i \in S^+} b_i(x) - \sum_{i \in S^+} \dis_i + \sum_{i \in T} b_i(x) + \sum_{i \in T^+} \mu_i & \\
= & \ds \sum_{i \in N \setminus S^-} b_i(x) - \sum_{i \in S^+} \dis_i + \sum_{i \in T^+} \mu_i &
\end{eqnarray*}
By definition, $x \in \dom v_i$ for $i \in T^+$, and we have $T^+ \subseteq N \setminus S^-$ since $S \cap T = \emptyset$. Therefore, $\sum_{i \in N \setminus S^-} b_i(x) \leq V(N \setminus S^-, T^+)$. Applying this inequality to the above, we continue:
\begin{eqnarray*}
& \ds \sum_{i \in N \setminus (S \cup T)} b_i(x) + \sum_{i \in S} \dbid_i(x) + \sum_{i \in T} \rbid_i(x) & \\
\leq & \ds V(N \setminus S^-, T^+) - \sum_{i \in S^+} \dis_i + \sum_{i \in T^+} \mu_i & \\
\leq & \ds V(N, \emptyset) - \sum_{i \in S^-} \dis_i - \sum_{i \in T^+} \rai_i - \sum_{i \in S^+} \dis_i + \sum_{i \in T^+} \mu_i & \\
= & \ds V(N, \emptyset) - \sum_{i \in S} \pi_i & \\
= & \ds \sum_{i \in N \setminus (S \cup T)} b_i(x^*) + \sum_{i \in S} \dbid_i(x^*) + \sum_{i \in T} \rbid_i(x^*) &
\end{eqnarray*}
Here the second inequality corresponds to core constraint~\eqref{eq:bicore-ub} for the pair of sets $(S^-, T^+)$, and the final equality was shown in~\eqref{eq:discounted-opt-value}. As $x^*$ was an arbitrary optimal allocation (with respect to the original bids), and $x$ was an arbitrary allocation, this shows that every optimal allocation remains optimal after the bid updates using discounts and raises from the bicore. This completes the proof.
\end{proof}
%

\paragraph{Corollaries}

Note that $\Vw(S) = V(S, \emptyset)$ and $\Vl(T) = V(N, N \setminus T)$. By the relationships in the following lemma, Proposition~\ref{prop:vcg-characterizations} and Theorem~\ref{thm:core-results} follow as corollaries from Theorem~\ref{thm:main-result}.
\begin{lemma} \label{lem:relationships}
The following relationships hold.
\medskip
    \begin{itemize}[wide=0.5em, leftmargin =*, nosep]
        \item $(\vdis, \vzero) \in \bicore(N, V) \Leftrightarrow$ $\vdis \in \core(N, \Vw)$.
        \item $(\vzero, \vrai) \in \bicore(N, V) \Leftrightarrow$ $\vrai \in \core(N, \Vl)$.
        \item $(\dis_i, \vzero_{-i}) \in \core(N, \Vw) \Leftrightarrow$ $\dis_i \in \mbtw_i$.
        \item $(\rai_i, \vzero_{-i}) \in \core(N, \Vl) \Leftrightarrow$ $\rai_i \in \mbtl_i$.
    \end{itemize}
\end{lemma}
\begin{proof}
A vector $(\vdis, \vzero) \in \bicore(N, V)$ satisfies $\sum_{i \in S} \pi_i \leq V(N, \emptyset) - V(N \setminus S, T)$ for all $S,T \subseteq N$. As $V$ is non-increasing in its second argument, this is equivalent to satisfying $\sum_{i \in S} \pi_i \leq V(N, \emptyset) - V(N \setminus S, \emptyset) = \Vw(N) - \Vw(N \setminus S)$ for all $S \subseteq N$, which are the constraints defining $\core(N, \Vw)$. The other relationships follow from analogous arguments.
\end{proof}

\subsection{Feedback Policies}

The core and bicore characterize joint bid updates that maintain the optimality of currently-optimal solutions. The analog of the minimum-bid-to-win (or maximum-bid-to-lose) from this perspective is to report \emph{maximal} elements from the core and bicore. Formally, we define the following three feedback policies.
\begin{description}
\item[VCG] The seller sends the VCG discount $\Vw(N) - \Vw(N-i)$ or the VCG raise $\Vl(N) - \Vl(N-i)$ to each agent $i \in N$ depending on whether it is winning or losing.
\item[Core] The seller sends discounts chosen from $\core(\Vw, N)$ that maximize $\sum_i \dis_i$, and raises chosen from $\core(\Vl, N)$ that maximize $\sum_i \rai_i$.
\item[Bicore] The seller sends discounts and raises chosen from $\bicore(V, N)$ that maximize $\sum_i \dis_i + \sum_i \rai_i$.
\end{description}
\medskip
Note that only VCG feedback prescribes unique discounts and raises. Maximal elements in the core and bicore are typically not unique. Let us return to the examples of Section~\ref{ex:vcg-examples} to see what kind of feedback comes out of the core and bicore policies.
\begin{example} {\em (Single-item auction)}
\ In the single-item auction of Example~\ref{ex:single-item-auction}, core feedback is equivalent to VCG feedback, but the bicore introduces some additional constraints. Suppose winning agent 1 and losing agent $i > 1$ both update their bids simultaneously. To ensure that agent 1 remains winning, the bicore imposes the constraint $\dis_1 + \rai_i \leq b_1 - b_i$, or equivalently $b_i + \mu_i \leq b_1 - \dis_1$. This shows that the current minimum-bid-to-win feedback policy used in display advertising is not a special case of bicore feedback.
\end{example}
\begin{example}  {\em (Zero-VCG)} \label{ex:zero-vcg-revisited}
\ In the video pod auction of Example~\ref{ex:zero-vcg}, core feedback would prescribe a raise of $\mu_1 = 10$ to the first agent, equivalent to its VCG raise. For the other two winning agents, however, it would prescribe discounts that satisfy $\dis_2 + \dis_3 \leq 10$. Therefore, it is no longer feasible for them to both apply their VCG discounts of 10 each. If we break ties in favor of the most uniform discount vector, core feedback would lead to $\pi_2 = \pi_3 = 5$.
Bicore feedback imposes even more constraints. To ensure that the current allocation remains optimal, the joint bid updates must satisfy $\rai_1 + \dis_2 + \dis_3 \leq 10$. Therefore, if we send raise $\mu_1 = 10$ to agent 1, we must send discounts $\pi_2 = \pi_3 = 0$ to agents 2 and 3.
\end{example}


\section{Assignment Problem}
\label{sec:assignemnt-problem}

In this section we consider the assignment problem, which is a special case of video pod auctions where agents have position-specific values, but the duration constraint is not binding and there are no exclusion constraints.
We give a linear programming characterization of the bicore in the assignment problem and show that it forms a lattice, extending classic results of~\citet{shapley1971assignment}. \citet{leonard1983elicitation} has shown that the largest element of this lattice corresponds to the VCG discounts for all agents at once. We show that the smallest element corresponds to the agents' VCG raises. Our treatment relies on several results presented in~\cite[Chapter 7.4]{vohra2004advanced}.

In the assignment problem the seller has a set of indivisible items $M$ to allocate to the agents. Let $m = M$. The bidding language is that each agent $i$ places a bid vector $\vb_i = (b_{i1}, \dots, b_{im})$ specifying bids for each item. Each agent can obtain at most one item, and each item goes to at most one agent. We have $m \leq n$ and assume without loss of generality that the seller always allocates all the items.

For $S \subseteq N$, let $\vd^S$ be the binary indicator for the set, where $d^S_i = 1$ if $i \in S$ and 0 otherwise. Let $S, T \subseteq N$. The following parametrized linear program evaluates $V(S,T)$ by definition.\footnote{By standard results on network flows, the feasible set has integer optimal solutions because the constraint matrix is totally unimodular; see e.g.~\cite[Chapter 4]{vohra2004advanced}.} 
\begin{eqnarray*}
\max_{\vx \geq 0} & \ds \sum_{i \in N} \sum_{j \in M} b_{ij} x_{ij} & \\
\mbox{s.t.} & \ds d_i^T \leq \sum_{j \in M} x_{ij} \leq d_i^{S \cup T} & \forall i \in N \\
 & \ds \sum_{i \in N} x_{ij} \leq 1 & \forall j \in M
\end{eqnarray*}
The dual linear program is as follows.
\begin{eqnarray*}
\min_{\vdis \geq 0, \vrai \geq 0, \vp \geq 0} & \ds \sum_{j \in M} p_j + \sum_{i \in N} d_i^{S \cup T} \dis_i - \sum_{i \in N} d_i^T \rai_i & \\
\mbox{s.t.} & \ds \dis_i - \rai_i \geq b_{ij} - p_j & \forall i \in N, \forall j \in M
\end{eqnarray*}
We refer to this dual as $D(S, T)$. Note that $D(N, \emptyset)$ is the dual to the allocation problem actually solved by the seller, defining the optimal allocation. Also note that the parametrization does not change the set of feasible dual solutions, only its objective.

\begin{theorem} \label{thm:bicore-lp}
If $(\vdis^*, \vrai^*, \vp^*)$ is an optimal solution to the dual $D(N, \emptyset)$, then $(\vdis^*, \vrai^*) \in \bicore(N, V)$. Conversely, if $(\vdis, \vrai) \in \bicore(N, V)$, then there is a vector $\vp \in \reals^m$ such that $(\vdis, \vrai, \vp)$ is an optimal solution to the dual $D(N, \emptyset)$.
\end{theorem}
%
%
\begin{appendixproof}
{\bf [Theorem~\ref{thm:bicore-lp}]}
Let $(\vdis^*, \vrai^*, \vp^*)$ be an optimal solution to $D(N, \emptyset)$, which is non-negative by the dual constraints. By strong duality, $\sum_{j \in M} p_j^* + \sum_{i \in N} \dis_i^* = V(N, \emptyset)$. Let $S , T \subseteq N$ be coalitions that satisfy $S \cap T = \emptyset$, which implies $N \setminus S \supseteq T$. As $(\vdis^*, \vrai^*, \vp^*)$ is feasible for $D(N \setminus S, T)$, we have $\sum_{j \in M} p_j^* + \sum_{i \in N \setminus S} \dis_i^* - \sum_{i \in T} \rai_i \geq V(N \setminus S, T)$. Combining the prior equality and inequality, we obtain $\sum_{i \in S} \dis_i^* + \sum_{i \in T} \rai_i^* \leq V(N, \emptyset) - V(N \setminus, S, T)$, which confirms that $(\vdis^*, \vrai^*) \in \bicore(N, V)$.

For the second part of the claim, let $(\vdis, \vrai) \in \bicore(N, V)$. As the bicore is downwards-closed, $(\vdis, \vzero) \in \bicore(N, V)$, and by Lemma~\ref{lem:relationships} we thus have $\vdis \in \core(\Vw, N)$. By Lemma 7.25 in~\citep{vohra2004advanced}, there is a $\vp \in \reals^m$ such that $(\vdis, \vzero, \vp)$ is an optimal solution to $D(N, \emptyset)$. We show that $(\vdis, \vrai, \vp)$ is feasible for $D(N, \emptyset)$, and therefore optimal because the objective only depends on $\vdis$ and $\vp$. If $i \in N$ is winning, then we have $\rai_i = 0$ as $(\vdis, \vrai)$ are valid discounts and raises, by Lemma~\ref{lem:core-valid}. Thus, for every $j \in M$, we have $\pi_i - \mu_i = \pi_i \geq b_{ij} - p_j$, by the feasibility of $(\vdis, \vzero, \vp)$. Now let $i \in N$ be a losing bidders, and take any $j \in M$. Let $\vx^*$ be an optimal solution to the primal, which defines an allocation, and let $k$ be the agent that wins item $j$ in this allocation. By complementary slackness, $\pi_k = b_{kj} - p_j$. We have $\sum_{i \in N} \sum_{j \in M} b_{ij} x^*_{ij} = V(N, \emptyset)$. If we switch agent $i$ for agent $k$ in the solution, the solution has value $\sum_{i \in N} \sum_{j \in M} b_{ij} x^*_{ij} + (b_{ij} - b_{kj}) \leq V(N-k,i)$, where the latter inequality holds by the definition of $V$. Combining the previous equality and inequality yields $V(N, \emptyset) - V(N-k, i) \leq b_{kj} - b_{ij}$. Combining this with the bicore inequality $\dis_k + \rai_i \leq V(N, \emptyset) - V(N-k, i)$, we have $\dis_k + \rai_i \leq  b_{kj} - b_{ij}$. Substituting in the complementary slackness equality $b_{kj} = \pi_k + p_j$, we finally obtain $\rai_i \leq p_j - b_{ij}$. This confirms the feasibility of $(\vdis, \vrai, \vp)$ and completes the proof.
\end{appendixproof}
%

For two $n$-vectors or $m$-vectors, let $\wedge$ be the operation of taking the component-wise min, and $\vee$ the operation of taking their component-wise max. For two solutions $(\vdis^1, \vrai^1, \vp^1)$ and $(\vdis^2, \vrai^2, \vp^2)$ to the dual, define the following meet and join operations:
\begin{itemize}
    \item $(\vdis^1, \vrai^1, \vp^1) \wedge (\vdis^2, \vrai^2, \vp^2) = (\vdis^1 \vee \vdis^2, \vrai^1 \wedge \vrai^2, \vp^1 \wedge \vp^2)$.
    \item $(\vdis^1, \vrai^1, \vp^1) \vee (\vdis^2, \vrai^2, \vp^2) = (\vdis^1 \wedge \vdis^2, \vrai^1 \vee \vrai^2, \vp^1 \vee \vp^2)$.
\end{itemize}

We have the following result. The main technical insight for the proof is that although the feasible set of $D(N, \emptyset)$ does not form a lattice, we can impose the constraint $\vdis \cdot \vrai = 0$ and obtain a lattice structure without excluding any optimal solutions.

\begin{theorem} \label{thm:assignment-lattice}
The set of optimal dual solutions to $D(N, \emptyset)$ forms a lattice under the meet and join operations $(\wedge, \vee)$ just defined. At the smallest solution, $\pi_i$ is agent $i$'s VCG discount for all $i \in N$. At the largest solution, $\mu_i$ is agent $i$'s VCG raise for all $i \in N$.
\end{theorem}
%

%
\begin{appendixproof}
{\bf [Theorem~\ref{thm:assignment-lattice}]}    
The proof is based on the results in~\cite[Chapter 7]{vohra2004advanced}. However, in our case the feasible set of the dual problem is not a lattice. We must first show that we can restrict our attention to a subset of the feasible set which does form a lattice.

For any parametrization $(S,T)$ of the dual, we first argue that we can impose the constraint $\pi_i \mu_i = 0$ for each $i \in N$ and not change the value of the optimal solution. (Note that this amounts to the constraint that at least one of $\pi_i$ and $\mu_i$ must be 0, for each $i \in N$.) To see this take any feasible solution, let $\delta_i = \min\{\pi_i, \mu_i\}$, and modify the solution by $\pi_i \leftarrow \pi_i - \delta_i$ and $\mu_i \leftarrow \mu_i - \delta_i$. Note that this modified solution remains feasible. The objective remains unchanged if $i \in T$ and otherwise weakly decreases. Therefore, the transformed solution is weakly better, and after transformation at least one of $\mu_i$ and $\pi_i$ is 0.

We make the change of variables $\lambda_i = -\pi_i$. Consider the feasible set
\begin{eqnarray*}
& \ds p_j \geq v_{ij} + \lambda_i + \mu_i & \forall i \in N, \forall j \in M \\
& \mu_i \lambda_i = 0 & \forall i \in N
\end{eqnarray*}
We show that this is a lattice under the operations
$$(\vrai^1, \vlambda^1, \vp^1) \wedge (\vrai^2, \vlambda^2, \vp^2) = (\min\{\vrai^1, \vrai^2\}, \min\{\vlambda^1, \vlambda^2\}, \min\{\vp^1, \vp^2\})$$
$$(\vrai^1, \vlambda^1, \vp^1) \vee (\vrai^2, \vlambda^2, \vp^2) = (\max\{\vrai^1, \vrai^2\}, \max\{\vlambda^1, \vlambda^2\}, \max\{\vp^1, \vp^2\})$$
where the $\min$ and $\max$ are understood component-wise. Consider two different elements of the feasible set, satisfying the inequalities:
$$ p^1_j \geq v_{ij} + \lambda^1_i + \mu^1_i $$
$$ p^2_j \geq v_{ij} + \lambda^2_i + \mu^2_i $$
Assume w.l.o.g.\ that $p^1_j \leq p^2_j$.
Then we have $(p^1_j \wedge p^2_j) \geq v_{ij} + \lambda^1_i + \mu^1_i$, and therefore
$(p^1_j \wedge p^2_j) \geq v_{ij} + (\lambda^1_i \wedge \lambda^2_i) + (\mu^1_i \wedge \mu^2_i)$. In the following case analysis, we thus only need to show that the constraint holds for the join, and that the $\mu_i \lambda_i = 0$ constraint (referred to as cs-constraint) holds for the meet and join.
\begin{itemize}
    \item {\bf Case 1: $\mu^1_i, \mu^2_i \geq 0$, $\lambda^1_i, \lambda^2_i = 0$.}
    Assume w.l.o.g.\ that $\mu^1_i \geq \mu^2_i$. Then we have $p^1_j \geq v_{ij} + (\mu^1_i \vee \mu^2_i)$ and therefore $(p^1_j \vee p^2_j) \geq v_{ij} + (\mu^1_i \vee \mu^2_i)$. The cs-constraint holds because $(\lambda^1_i \wedge \lambda^2_i) = (\lambda^1_i \vee \lambda^2_i) = 0$.
    \item {\bf Case 2: $\mu^1_i, \mu^2_i = 0$, $\lambda^1_i, \lambda^2_i \leq 0$.} By the symmetry between $\mu_i$ and $\lambda_i$ in the feasibility constraints, this is equivalent to case 1.
    \item {\bf Case 3: $\mu^1_i \geq 0, \mu^2_i = 0$, $\lambda^1_i = 0, \lambda^2_i \leq 0$.} Here $\mu^1_i \vee \mu^2_i = \mu^1_i$ and $\lambda^1_i \vee \lambda^2_i = \lambda^1_i = 0$. Therefore, we have $p^1_j \geq v_{ij} + \lambda^1_i + \mu^1_i = v_{ij} + (\mu^1_i \vee \mu^2_i) + (\lambda^1_i \vee \lambda^2_i)$, leading to $(p^1_j \vee p^2_j) \geq v_{ij} + (\mu^1_i \vee \mu^2_i) + (\lambda^1_i \vee \lambda^2_i)$.
    \item {\bf Case 4: $\mu^1_i \geq 0, \mu^2_i = 0$, $\lambda^1_i \leq 0, \lambda^2_i = 0$.} Because of the constraint that $\mu_i \lambda_i = 0$, this reduces to either case 1 or 2.
\end{itemize}

\medskip
Let $X$ denote the feasible set
\begin{eqnarray*}
& \ds p_j \geq v_{ij} + \lambda_i + \mu_i & \forall i \in N, \forall j \in M \\
& \mu_i \lambda_i = 0 & \forall i \in N \\
& \vrai \geq 0,\, \vlambda \leq 0,\, \vp \geq 0 &
\end{eqnarray*}
which we now know forms a lattice under the component-wise max (join) and min (meet) operations, and also contains the optimal solution to the original dual problem. Let $Y = \{0,1\}^N$ be the lattice from which $d^T$ is drawn, with the same meet and join operations. Let $f: X \times Y \rightarrow \mathbf{R}$ be the objective:
$$
f(\vrai, \vlambda, \vp, \vd) = \sum_{j \in M} p_j - \sum_{i \in N} \lambda_i - \sum_{i \in N} d_i \mu_i
$$
This objective is submodular. Therefore $h: Y \rightarrow \mathbf{R}$ defined as $h(y) = \min_{x \in X} f(x,y)$ is also submodular, and its minimizers form a lattice~\cite[Theorem 7.5]{vohra2004advanced}. As $\Vl(T) = h(\vd^{N \setminus T})$, it is also submodular. The fact that $\Vw$ is submodular is proven in~\cite[Theorem 7.21]{vohra2004advanced}. This proves the first claim. 

From this lattice structure, Theorem~\ref{thm:bicore-lp}, and by inspection of~\eqref{eq:bicore-ub}, we see that smallest solution (with largest $\vdis$) has the form $(\vdis, \vzero)$, so $\vdis \in \core(N, \Vw)$ by Lemma~\ref{lem:relationships}. Similarly, the largest solution (with largest $\vrai$) has the form $(\vzero, \vrai)$, so $\vrai \in \core(N, \Vl)$, again by Lemma~\ref{lem:relationships}. The final claim then follows from~\cite[Theorem 7.24]{vohra2004advanced} and the submodularity of $\Vw$ and $\Vl$.
\end{appendixproof}
%

By this theorem and the preceding Theorem~\ref{thm:bicore-lp}, it follows that $\bicore(N, V)$ also has a lattice structure in the assignment problem, and combined with Lemma~\ref{lem:relationships}, so do $\core(N, \Vw)$ and $\core(N, \Vl)$. As part of the proof, we show that $\Vl$ is submodular, in analogy to the fact that $\Vw$ is submodular~\cite[Theorem 7.21]{vohra2004advanced}.


\section{Bidding Dynamics}
\label{sec:bidding-dynamics}

In the previous sections we developed three alternative policies for computing generalized minimum-bid-to-win feedback: VCG, core, and bicore feedback. How should they be evaluated against each other? To compare them, we propose to empirically examine the bidding dynamics they induce if agents follow the feedback. We are interested in verifying whether the dynamics converge, and if so, how quickly they converge and whether they converge to an efficient allocation.

We consider a video pod setting where each agent $i$ has a uniform value $v_i$ across positions and an ad duration.
To define the dynamics we assume that each agent plays a $\rho_i$-profit-target strategy in each round, where $\rho_i \in [0, v_i]$, and the auction provides valid discounts $\vdis$ and raises $\vrai$ as feedback. Let $\epsilon$ be a small bid increment. Recall that a decrease in the profit target corresponds to an increase in the bid, and vice-versa. The bidding dynamics are defined as follows:
\begin{itemize}
    \item A \emph{winning} agent $i$ updates its target according to: $\rho_i \leftarrow \min\{ \rho_i + \dis_i - \epsilon, v_i \}$.
    \item A \emph{losing} agent $i$ updates its target according to:  $\rho_i \leftarrow \max\{ \rho_i - \rai_i - \epsilon, 0 \}$.
\end{itemize}
The bid increment $\epsilon$ is important to avoid stalls in the dynamics. For instance, suppose all agents initially place a bid of 0. Then every allocation is trivially optimal, and all agents are both winning and losing. We would therefore have $\vdis = \vrai = \vzero$ with no progress in the bidding. To guard against ties of this sort, each agent should apply a small increment of $\epsilon$ to its bid updates (equivalently, decrease its profit-target by $\epsilon$) in an effort to ensure that it wins in every optimal allocation. In a single-item auction (Example~\ref{ex:single-item-auction}), it is straightforward to check that these dynamics lead to a standard ascending-price auction with bid increment of $\epsilon$ when starting from bids of 0. 

\begin{example} \label{ex:zero-vcg-dynamics}
To give an illustration of how the dynamics might behave, Table~\ref{tab:zero-vcg-dynamics} shows their evolution for the auction instance given in Example~\ref{ex:zero-vcg}, assuming that bidders initially bid their true values.  We observe that the VCG dynamics cycle between the efficient allocation to agents 2 and 3 and an inefficient allocation to agent 1. The reason is that VCG feedback induces agents to "overshoot" with their bid updates, leading to oscillating bids. Under core feedback, on the other hand, agents 2 and 3 coordinate their bid discounts so that they jointly still exceed the bid of agent 1 in the second round, leading to convergence. Note, however, that the core feedback for the first round is not unique. Another possible update would have led to bids (10, 10, $\epsilon$) in the second round, which is another possible endpoint of the dynamics.
\begin{table}[ht]
\centering
\begin{tabular}{l|rrr|rrr|rrr}
   \multicolumn{1}{c}{} & \multicolumn{3}{c}{VCG} & \multicolumn{3}{c}{Core} & \multicolumn{3}{c}{Bicore}  \\
   \toprule
   Round & Bid 1 & Bid 2 & Bid 3 & Bid 1 & Bid 2 & Bid 3 & Bid 1 & Bid 2 & Bid 3 \\
   \midrule
   1     & 10 & {\bf 10} & {\bf 10} & 10 & {\bf 10} & {\bf 10} & 10 & {\bf 10} & {\bf 10} \\
   2     & {\bf 10} & $\epsilon$ & $\epsilon$ & 10 & {\bf 5} + $\bm{\epsilon}$ & {\bf 5} + $\bm{\epsilon}$ & 10 & {\bf 10} & {\bf 10} \\
   3     & 3$\epsilon$ & {\bf 10} & {\bf 10} & 10 & {\bf 5} + $\bm{\epsilon}$ & {\bf 5} + $\bm{\epsilon}$ & & & \\
   4     & {\bf 10} & $\epsilon$ & $\epsilon$ & & & & & \\  
  \bottomrule
\end{tabular}
\medskip
\caption{Evolution of the bidding dynamics for Example~\ref{ex:zero-vcg} under the three feedback policies. Bold bids indicate the winning agents at each round.}
\label{tab:zero-vcg-dynamics}
\end{table}
\end{example}
\noindent
Bicore feedback is also not unique. If raise $\rai_1 = 10$ is sent to agent 1, then as explained in Example~\ref{ex:zero-vcg-revisited} the only possible discounts that leave the optimal allocation unchanged, were agent 1 to apply the recommended raise, are $\dis_2 = \dis_3 = 0$. If we send this feedback, however, the bids do not change because agent 1 is already bidding its value, and the dynamics have already converged by round 2.

\section{Empirical Analysis}
\label{sec:empirical-analysis}

 In this section we report on simulations of the bidding dynamics described in Section~\ref{sec:bidding-dynamics}.
The dataset for our simulations consists of 10,000 video pod auctions sampled randomly from the logs of an ad exchange. For each auction we have the number of positions, maximum total ad duration, and information on the bidders: their uniform value per position, and their ad duration. We restricted our attention to auctions with 3--5 bidders where the constraints on the number of ads and total ad duration were binding, to make sure they represented nontrivial allocation problems. We did not consider auctions with exclusion constraints, or where bidders had specific values per position. We used the agents' bids recorded in the logs as their values for the sake of simulation. Each agent $i$'s initial profit target was drawn uniformly from $[0, v_i]$, and its bid increment was set to $v_i/10$.

There are three termination criteria: convergence, cycling, or max-rounds. We imposed a cap of 20 rounds on each auction. We consider that bidding dynamics have converged if the change in $\sum_i b_i$ is less than 1\% from one round to the next. We consider that they have cycled if the agents place bids that are identical to bids from a previous round, not including the immediate prior round, and if the change in $\sum_i b_i$ from the prior round was at least 10\%.
Core and bicore feedback was obtained by solving linear programs.

The simulation results are presented in Table~\ref{tab:simulation-results}. For each combination of feedback policy and number of bidders, the table first shows the average number of rounds and average efficiency upon termination. The next three columns give the frequencies of the three termination criteria, in percentage terms. The last column gives the average efficiency of the allocation when a cycle is detected---cycles only occurred under VCG feedback.

We see that core and bicore feedback led to convergence in over 99\% of the instances, and never cycled. VCG feedback similarly has high convergence rate, but with 3 bidders it led to cycling in over 8\% of the instances. The frequency of cycling drops sharply with the number of bidders, but the allocations that are reached upon cycling are more inefficient.

\begin{table}[tbh]
\centering
\small
\vspace{10pt}
\begin{tabular}{lcrrrrrr}
  \toprule
Policy & Bidders & Avg. Rounds & Avg. Eff. & Conv. & Cycle & Max-Rnds & Avg. Cycle Eff. \\ 
  \midrule
   VCG & 3 & 5.60 (0.03) & 99.50 (0.05) & 90.44 & 8.42 & 1.14 & 96.23 (0.46) \\ 
   CORE & 3 & 5.40 (0.03) & 100 (0.00) & 99.90 & 0.00 & 0.10 & \color{gray} n/a \\ 
   BICORE & 3 & 7.11 (0.03) & 100 (0.00) & 99.88 & 0.00 & 0.12 &  \color{gray} n/a \\ 
   \midrule
   VCG & 4 & 5.38 (0.05) & 99.82 (0.05) & 98.38 & 1.27 & 0.35 & 88.00 (3.31) \\ 
   CORE & 4 & 5.31 (0.04) & 100 (0.00) & 99.91 & 0.00 & 0.09 &  \color{gray} n/a \\ 
   BICORE & 4 & 6.77 (0.06) & 100 (0.00) & 99.09 & 0.00 & 0.91 &  \color{gray} n/a \\ 
   \midrule
   VCG & 5 & 4.97 (0.05) & 99.97 (0.03) & 99.83 & 0.17 & 0.00 & 83.22 (6.34) \\ 
   CORE & 5 & 4.95 (0.05) & 100 (0.00) & 99.91 & 0.00 & 0.09 &  \color{gray} n/a \\ 
   BICORE & 5 & 5.58 (0.07) & 100 (0.00) & 99.64 & 0.00 & 0.36 &  \color{gray} n/a \\ 
   \bottomrule
\end{tabular}
\medskip
\caption{Simulation results. All metrics except for average rounds are reported in percentage terms. For averages, standard errors are given in parentheses.} 
\vspace{10pt}
\label{tab:simulation-results}
\end{table}

\noindent
All feedback policies converge in just a few rounds on average, with bicore taking 1--2 more rounds than the other two policies on average. Core and bicore reach the max-rounds cap of 20 in less than 1\% of the instances, and in these cases the final allocation was always efficient.

\begin{figure}[t!]
        \centering
	\includegraphics[scale=0.3]{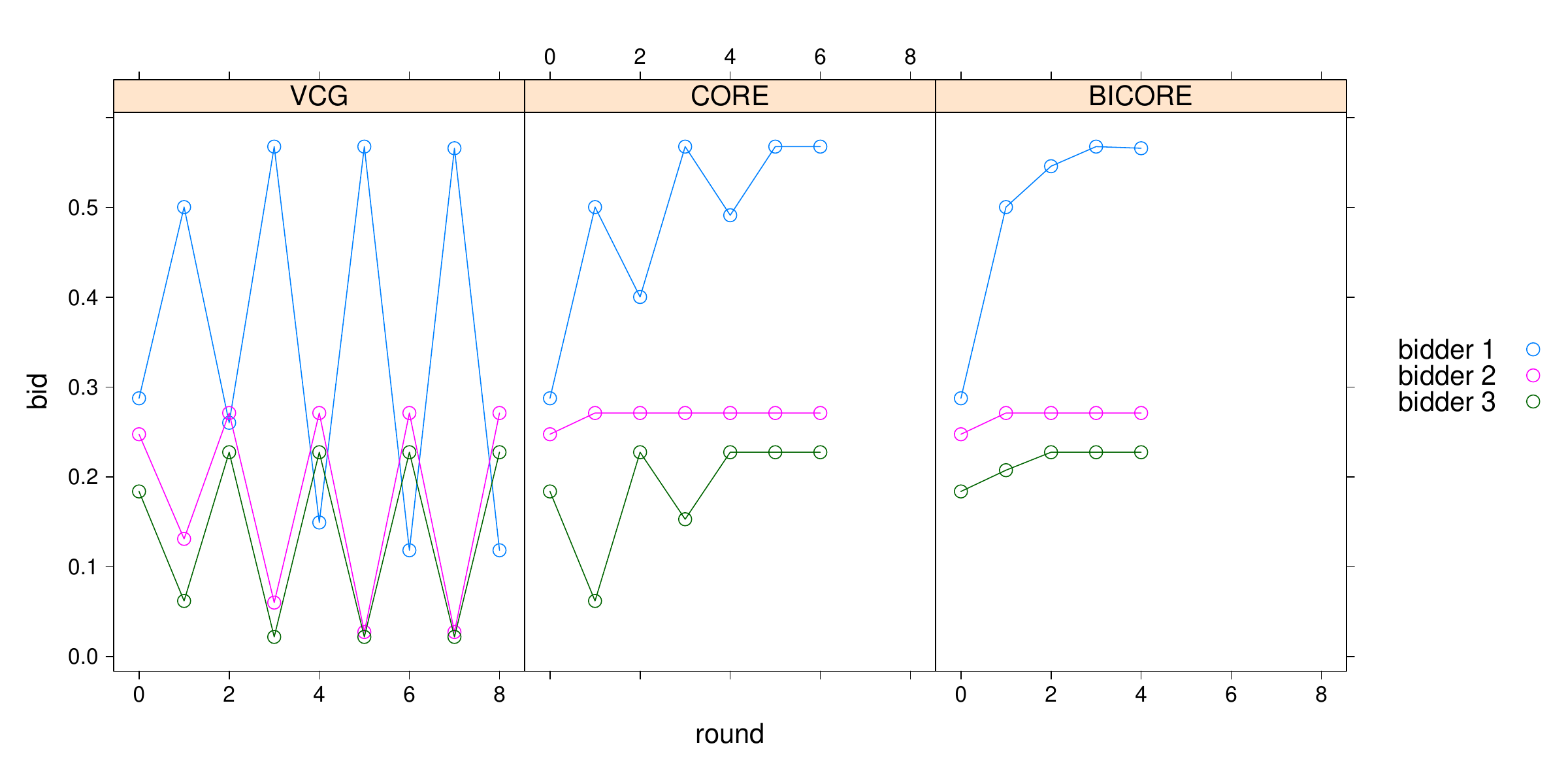}
        \vspace{-10pt}
	\caption{Instance of cycling under VCG feedback in the simulations. Core and bicore feedback both converge to the same efficient endpoint.}
	\label{fig:cycling-example}
\end{figure}
Figure~\ref{fig:cycling-example} plots the evolution of bids for an instance where VCG feedback led to cycling. Because the plot is based on confidential data, the y-axis has been rescaled by a constant. We see in this realistic instance the same oscillation phenomenon described earlier in Example~\ref{ex:zero-vcg-dynamics}. Here core and bicore feedback converge to the same endpoint, but there are some bid oscillations under core feedback whereas the bid updates are closer to monotone under bicore feedback, due to the fact that bid updates are coordinated across all agents (both winning and losing).


\section{Conclusion}
\label{sec:conclusion}

This work proposed three generalizations of the minimum-bid-to-win feedback provided in single-slot display ad auctions to general combinatorial auctions, motivated by the concrete use case of video pods. 
Our approach was to characterize the sets of bid updates (discounts for winning bidders, raises for losing bidders) that maintain optimality of the current optimal allocation. A generalized minimum-bid-to-win is then any maximal vector under this characterization. For updates by a single agent at a time, this led us to define the concept of the VCG raise, in analogy to the VCG discount. For updates by multiple winning (or losing) bidders at once, we provided characterizations based on the core of associated cooperative games. For our main result, we proposed a core concept for bicooperative games called the bicore, and showed that it completely characterizes the bid updates that maintain optimality of the current optimal allocation for any mix of winning and losing bidders.

For the assignment problem---a special case of video pod auctions---we provided a linear programming characterization of the bicore, providing an efficient way to compute feedback for bidders. It was previously known that bidder-optimal core payoffs, associated with minimal clearing prices, correspond to VCG discounts in the assignment problem~\citep{leonard1983elicitation}. Our results show that maximal clearing prices are associated with the bidders' VCG raises, providing a new interpretation of the seller-optimal payoff point. There are linear programming characterizations of the core for several well-known optimization problems, including facility location and linear production models~\citep{goemans2004cooperative,owen1975core}. Extending these characterizations to the bicore, if possible, may offer additional structural insights.


There are two practical challenges to computing core and bicore feedback: they involve a large number of constraints, and they require solving multiple allocation problems to obtain the coalitional values. The three generalizations we investigated are in fact part of a spectrum, with VCG and bicore being the two extremes. Any subset of the bicore constraints leads to a feedback policy. Taking a data-driven approach to select a small subset of core or bicore constraints likely to bind may work well in practice~\citep{balcan2022structural}. For the coalitional values, several allocation problems need to be solved even for VCG, but the problems are very closely related. Approximations of these values could be obtained during a search of the allocation space in a branch-and-bound algorithm, for instance, and it may be possible to tune the search to explore the most relevant coalitions. This is a promising avenue for future work.



\bibliographystyle{abbrvnat}
\bibliography{mbtw-arxiv}

\end{document}